\documentclass{sig-alt-full}
\conferenceinfo{ISSAC'08,} {July 20--23, 2008, Hagenberg, Austria.}
\CopyrightYear{2008}
\crdata{978-1-59593-904-3/08/07}

\usepackage{amsmath,amssymb}
\usepackage{theorem}
\usepackage{graphics}

\usepackage{url}
\usepackage[plainpages=false,pdfpagelabels,colorlinks=true,citecolor=blue,hypertexnames=false]{hyperref}

\newcommand{\x}{X}
\newcommand{\y}{Y}
\newcommand{\Tx}{\theta}
\newcommand{\Dx}{\partial}

\newcommand{\bigO}{{\mathcal{O}}}
\newcommand{\bigOsoft}{\tilde{\mathcal{O}}}
\newcommand{\MM}{\mathsf{MM}}
\def\OMul#1#2#3{\langle #1,#2 \rangle_{#3}}
\def\MMul#1#2#3{\langle #1,#2,#3 \rangle}

\newcommand{\sC}{\mathsf{C}}
\newcommand{\sM}{\mathsf{M}}
\newcommand{\sT}{\mathsf{T}}
\newcommand{\bK}{\mathbb{K}}
\newcommand{\bN}{\mathbb{N}}
\newcommand{\bQ}{\mathbb{Q}}
\newcommand{\bZ}{\mathbb{Z}}
\newcommand{\sF}{\mathsf{F}}

\def\gathen#1{{#1}}
\def\hoeven#1{{#1}}

\newtheorem{lemma}{Lemma}
\newtheorem{theorem}{Theorem}
\newtheorem{proposition}{Proposition}

\def\tM{\tilde M}

\begin{document}

\title{Products of Ordinary Differential Operators by Evaluation~and~Interpolation}
\numberofauthors{3}

\author{
\alignauthor Alin Bostan\\
\affaddr{Algorithms Project-Team, INRIA Paris-Rocquencourt}\\
\affaddr{78153 Le Chesnay (France)}\\
\email{Alin.Bostan@inria.fr}
\alignauthor Fr\'ed\'eric Chyzak \\
\affaddr{Algorithms Project-Team, INRIA Paris-Rocquencourt}\\
\affaddr{78153 Le Chesnay (France)}\\
\email{Frederic.Chyzak@inria.fr}
\alignauthor Nicolas Le Roux\\
\affaddr{Algorithms Project-Team, INRIA Paris-Rocquencourt}\\
\affaddr{78153 Le Chesnay (France)}\\
\email{Nicolas.Le\_Roux@inria.fr}
}

\date{\today}
\maketitle

\begin{abstract}
It is known that multiplication of linear differential operators over ground fields of characteristic zero can be reduced to a constant number of matrix products.
We give a new algorithm by evaluation and interpolation which is faster than the previously-known one by a constant factor, and
prove that in characteristic zero, multiplication of differential operators and of matrices are computationally equivalent problems.
In positive characteristic, we show that differential operators can be multiplied in nearly optimal time. 
Theoretical results are validated by intensive experiments.
\end{abstract}

\vspace{1mm}
 \noindent
 {\bf Categories and Subject Descriptors:} \\
\noindent I.1.2 [{\bf Computing Methodologies}]:{~} Symbolic and Algebraic
  Manipulation -- \emph{Algebraic Algorithms}
 
 \vspace{1mm}
 \noindent
 {\bf General Terms:} Algorithms, Theory
 
 \vspace{1mm}

\section{Introduction}

Multiplication in polynomial algebras $\bK[\x]$ and $\bK[\x,\y]$ over a field~$\bK$ has been intensively studied in the computer-algebra literature.
Since the discovery of Karatsuba's algorithm and the Fast Fourier Transform, hundreds of articles have been dedicated to theoretical and practical issues;
see, e.g., \cite[Ch.~8]{GaGe99}, \cite{Bernstein}, and the references therein.
Not only are many other operations built upon multiplication, but often their complexity can be expressed in terms of the complexity of multiplication---whether as a constant number of multiplications or a logarithmic number of multiplications.
In~$\bK[\x]$, this is the case for Euclidean division, gcd and resultant computation, multipoint evaluation and interpolation, shifts, certain changes of bases, etc.

In the noncommutative setting of linear ordinary differential operators, the study is by far less advanced.
The complexity of the product has been addressed only recently, by van~der~Hoeven in the short paper~\cite{vdHoeven02}:
multiplication of operators over ground fields~$\bK$ of characteristic zero can be reduced by an evaluation-interpolation scheme to a constant number~$C$ of matrix multiplications with elements in~$\bK$.
Work in progress~\cite{LCLMs} suggests that linear algebra is again the bottleneck for computations of GCRDs and LCLMs.

This work aims at deepening the study started in~\cite{vdHoeven02} for characteristic~0.
We improve van der Hoeven's result along several directions:
We make the constant factor~$C$ explicit in~\S\ref{ssec:MatToOper} and improve it in~\S\ref{sec:better-constants}, and we prove in~\S\ref{sec:equiv} that multiplication of matrices and of differential operators are equivalent computational problems---that is, they share the same exponent, 
thus answering the question left open in~\cite[\S6, Remark~2]{vdHoeven02}.
As usual, those results hold for sufficiently large characteristic as well.
We prolong the study to the case of (small) positive characteristic, by giving in~\S\ref{sec:positive-char} an algorithm for computing the product of two differential operators in softly quadratic complexity, that is, nearly optimally in the output size.
This indicates that the equivalence result may fail to generalize to arbitrary fields.

In what follows, the field $\bK$ has characteristic zero, unless stated otherwise.
$\bK[\x]\langle\Dx\rangle$ and $\bK[\x]\langle\Tx\rangle$ respectively denote the associative algebras $\bK\langle\x,\Dx;\Dx\x=\x\Dx+1\rangle$ and $\bK\langle\x,\Tx;\Tx\x= \x(\Tx+1)\rangle$.

\vspace{-0.3cm}
\begin{table}[ht]
\begin{center}
\setlength{\tabcolsep}{3pt}
\begin{tabular}[t]{c|cc|ccc}
 & \textsf{vdH${}_\Tx$}  & \textsf{IvdH${}_\Tx$} & \textsf{vdH} & \textsf{IvdH} & \textsf{MulWeyl}\\
\hline
Product by blocks & 37 & 24 & 96 & 48 & 12 \\
\hline
Zeros + Strassen  & 20 &  8 & 47 & 12 &  8
\end{tabular}
\caption{\label{table:MM} Number of $n \times n$ matrix products for multiplication in $\bK[\x]\langle\Tx\rangle$, resp.~$\bK[\x]\langle\Dx\rangle$, in bidegree $(n,n)$.}
\end{center}
\end{table}

\vspace{-0.5cm}
Table~\ref{table:MM} encapsulates our improvements on the constant~$C$.
It displays the cost of linear algebra in van der Hoeven's algorithms (\textsf{vdH${}_\Tx$}, resp.~\textsf{vdH})
and in the improved versions (\textsf{IvdH${}_\Tx$}, resp.~\textsf{IvdH}), which are described in~\S\ref{ssec:OperToMat}, resp.~\S\ref{ssec:vdH_D}, and in our algorithm (\textsf{MulWeyl}) in~\S\ref{sec:MulWeyl}.
The subscript~$\Tx$ refers to multiplication in $\bK[\x]\langle\Tx\rangle$;
its absence means a product in $\bK[\x]\langle\Dx\rangle$. 
The first row provides bounds on the number of $n \times n$ matrix products used in each algorithm for multiplying operators in $\bK[\x]\langle\Dx\rangle$, resp.~$\bK[\x]\langle\Tx\rangle$, of degree at most $n$ in $\x$ and in $\Dx$, resp.~$\Tx$,
under the naive complexity estimate~\eqref{eq:naive-abc} below. 
This estimate reflects the choice of multiplying rectangular matrices by decomposing them into square blocks.
The second row gives tighter bounds under the assumptions that: \emph{(i)\/} any product by a zero block is discarded; \emph{(ii)\/} when possible, a product of two $2\times2$ matrices of $n\times n$ blocks is computed as 7~block products, instead of~8, by using Strassen's algorithm~\cite{Strassen69};
\emph{(iii)\/} predicted non-trivial zero blocks in the output are not computed.

\paragraph*{Canonical form and bidegree}
In the algebra $\bK[\x]\langle\Dx\rangle$, resp.~$\bK[\x]\langle\Tx\rangle$, the commutation rule allows one to rewrite any given element into a so-called \emph{canonical form\/} with $\x$~on the left of monomials and $\Dx$, resp.~$\Tx$, on the right, that is, as a linear combination of monomials~$\x^i\Dx^j$, resp.~$\x^i\Tx^j$, for uniquely-defined coefficients from~$\bK$.
In either case, we speak of an element of bidegree $(d,r)$, resp.~at most $(d,r)$, when the degree of its canonical form in~$\x$ is~$d$, resp.~at most~$d$, and that in~$\Dx$, resp.~$\Tx$, is~$r$, resp.~at most~$r$.
With natural notation, the bidegree $(d_C,r_C)$ of a product~$C=BA$ clearly satisfies $r_C=r_A+r_B$ and $d_C\leq d_A + d_B$.

The problem of computing the canonical form of the product of two elements of bidegree~$(d,r)$ from $\bK[\x]\langle\Dx\rangle$, resp.\ from $\bK[\x]\langle\Tx\rangle$, given in canonical form, is denoted $\OMul dr\Dx$, resp.\ $\OMul dr\Tx$.

\paragraph*{Complexity measures}
All complexity estimates are given in terms of arithmetical operations in~$\bK$, which we denote ``ops.''
We denote by $\sC_\Tx, \sC_\Dx:\bN \rightarrow \bN$ two functions such that Problems $\OMul nn\Dx$ and~$\OMul nn\Tx$ can be solved in $\sC_\Dx(n)$ and $\sC_\Tx(n)$, respectively.
We denote by $\sM: \bN \rightarrow \bN$ a function such that polynomials of degree at most~$n$ in~$\bK[\x]$ can be multiplied in $\sM(n)$~ops.
Using Fast Fourier Transform algorithms, $\sM(n)$ can be taken in $\bigO(n \log n)$ over fields with suitable roots of unity, and $\bigO(n \log n\,\log\log n)$ in the general case~\cite{ScSt71,CaKa91}. 
We use the notation $f \in \bigOsoft(g)$ for $f,g:\bN \rightarrow \bN$ if $f$ is in $\bigO(g \log^mg)$ for some $m\geq 1$.
For instance, $\sM(n)$ is in $\bigOsoft(n)$.
The problem of multiplying an $m \times n$ matrix by an $n \times p$ matrix is written $\MMul mnp$. 
We let $\MM: \bN^3 \rightarrow \bN$ be a function such that Problem  $\MMul mnp$ can be solved in $\MM(m,n,p)$~ops.
We use the abbreviation $\MM(n)$ for $\MM(n,n,n)$.
The current tightest (strict) upper bound 2.376 for~$\omega$ such that $\MM(n)\in\bigO(n^\omega)$ is derived in~\cite{CoWi90}.
For the time being, this estimate is only of theoretical relevance.
Few practical algorithms with complexity better than cubic are currently known  for matrix multiplication, among which
Strassen's algorithm~\cite{Strassen69} with exponent $\log_27 \approx 2.807$ and the Pan--Kaporin algorithm~\cite{Kaporin04} with exponent $2.776$.
For rectangular matrix multiplication,
we shall use the estimate
\begin{equation}\label{eq:naive-abc}
\MM(an,bn,cn) \leq abc \, \MM(n),\qquad\text{for $a,b,c \in \bN$,}
\end{equation}
obtained by performing the naive product of $a\times b$ by $b\times c$ matrices whose coefficients are $n\times n$ blocks.

Furthermore, we assume that $\sM(n)$, $\MM(n)$,  $\sC_\Dx(n)$, and $\sC_\Tx(n)$ satisfy the usual super-linearity assumption of~\cite[\S8.3, Eq.~(9)]{GaGe99} and also that, if $\sF(n)$ is any of these functions,  then $\sF(cn)$ belongs to $\bigO\bigl(\sF(n)\bigr)$, for all positive constants~$c$.

\paragraph*{Useful complexity results}
Throughout, we shall freely use several classical results on the complexity of basic polynomial operations.
They are encapsulated in Lemma~\ref{cost-results}.
The corresponding algorithms are found in: \cite[Algorithm~E]{GaGe97} for~(a); \cite[Chapter~10]{GaGe99} for~(b); \cite[Th.~2.4 and 2.5]{Gerhard00} for~(c); and \cite[Cor.~8.29]{GaGe99} for~(d).

\begin{lemma}\label{cost-results}
Let\/ $\bK$ be an arbitrary field. Let $a \in \bK$, let $P(\x) \in \bK[\x]$ be of degree less than $n$ and $f, g \in \bK[\x,\y]$ of degree at most $d$ in $\x$ and $n$ in $\y$.
One can perform:
\emph{(a)\/} the Taylor shift $Q(\x):=P(\x+a)$;
\emph{(b)\/} the multipoint evaluation and interpolation of $P$ on $a,a+1,\ldots,a+n$ if the characteristic of\/ $\bK$ is 0 or greater than $n$;
\emph{(c)\/} the base change between the monomial and  the falling factorial basis $(\x)_k=\x(\x-1)\cdots (\x-k+1)$
in $\bigO\bigl(\sM(n) \log n\bigr)$~ops.
Moreover, one computes:
\emph{(d)\/} the product $h=fg$ in $\bigO\bigl(\sM(dn)\bigr)$~ops.
\end{lemma}

\section{Naive algorithms}
\label{sec:naive-algos}

In this section, we provide complexity estimates for several known algorithms for~$\OMul dr\Dx$.
We set
\begin{equation*}
A=\sum_{i=0}^r\sum_{j=0}^da_{i,j}\x^j\Dx^i,\ \
B=\sum_{i=0}^r\sum_{j=0}^db_{i,j}\x^j\Dx^i=\sum_{i=0}^rb_i(\x)\Dx^i.
\end{equation*}
For any~$L=\sum_{i=0}^rl_i(\x)\Dx^i=\sum_{j=0}^d\x^jl'_j(\Dx)$, we define
\begin{equation*}
\frac{dL}{d\x} = \sum_{i=0}^r\frac{dl_i(\x)}{d\x}\Dx^i,\qquad
\frac{dL}{d\Dx} = \sum_{j=0}^d\x^j\frac{dl'_j(\Dx)}{d\Dx}.
\end{equation*}

\paragraph*{Naive expansion}

The most naive calculation of~$BA$ is by expanding each~$\Dx^i\x^l$ in the equality
\begin{equation*}
BA = \sum_{i=0}^r\sum_{j=0}^d\sum_{k=0}^r\sum_{l=0}^db_{i,j}a_{k,l}\x^j\left(\Dx^i\x^l\right)\Dx^k.
\end{equation*}
 Using Leibniz's formula $\Dx^i\x^l = \sum_{k=0}^{\min(i,l)}(l)_k\binom ik \x^{l-k} \Dx^{i-k}$ and the recurrences $(l)_{k+1} = (l)_k (l-k)$ and $\binom i{k+1} = \binom ik  \frac{i-k}{k+1}$, the canonical form of $\Dx^i\x^l$ is computed in $\bigO\bigl(\min(i,l)\bigr)$~ops. 
This induces a complexity $\bigO\bigl(d^2r^2\min(d,r)\bigr)$ for computing~$BA$.
The estimate simplifies to $\bigO(n^5)$ if~$d=r=n$.

\paragraph*{Iterative schemes}
Another calculation is by the formula
\begin{equation}\label{eq:iterative}
BA = \sum_{i=0}^rb_i(\x)\left(\Dx^iA\right)
\end{equation}
and the observation that $\Dx^iA$~has bidegree at most $(d,r+i)$ and is computed from~$\Dx^{i-1}A$ in $\bigO(dr)$~ops.\ by the identity
\begin{equation}\label{eq:DxT}
\Dx T=T\Dx+\frac{dT}{d\x}\qquad\text{for $T=T(\x,\Dx)$}.
\end{equation}
Therefore, the overall complexity is $\bigO\bigl(\sM(d)r^2+dr^2\bigr)=\bigO\bigl(\sM(d)r^2\bigr)$.
When $d=r=n$, this is~$\bigO\bigl(\sM(n)n^2\bigr)$, and $\bigOsoft\bigl(n^3\bigr)$ if FFT is used.
Similar considerations based on
\begin{equation}\label{eq:Tx}
T\x=\x T+\frac{dT}{d\Dx}\qquad\text{for $T=T(\x,\Dx)$}
\end{equation}
provide an algorithm in $\bigO\bigl(d^2\,\sM(r)\bigr)$, and one can always use the better algorithm by first comparing $d$ and~$r$.

Another formula, attributed to Takayama and used in several implementations (Takayama's \textsf{Kan} system~\cite{Kan}; Maple's \textsf{Ore\_algebra} by Chyzak~\cite{Ore_algebra}), is given by the (finite) sum
\begin{equation}\label{eq:takayama}
BA=\sum_{k\geq0}\frac1{k!}\left(\frac{d^kB}{d\Dx^k}\ast\frac{d^kA}{d\x^k}\right),
\end{equation}
where the products~$\ast$ are computed formally as commutative products between canonical forms, the resulting sum being viewed as a canonical form.
Each of the derivatives has bidegree at most $(d,r)$ and the derivative at order $k$ can be computed in $\bigO(dr)$~ops.\ from the one at order~$k-1$.
The complexity is seen to be $\bigO\bigl(\min(d,r)\,\sM(dr)\bigr)$~ops., by Lemma~\ref{cost-results}(d).
When $d=r=n$, this is~$\bigO\bigl(n\,\sM(n^2)\bigr)$, or $\bigOsoft(n^3)$ using FFT;
the scheme~\eqref{eq:iterative} is just a bit better than~\eqref{eq:takayama}.

\section{Equivalence between products of matrices and operators}
\label{sec:equiv}
Let $\bK$ be a field of characteristic zero.
In~\cite{vdHoeven02}, van der Hoeven showed that $\sC_\Tx(n)$ and $\sC_\Dx(n)$ are in $\bigO\bigl(\MM(n)\bigr)$.
When $\omega < 3$, this improves upon the algorithms in~\S\ref{sec:naive-algos}.

In this section, we explain and improve this result along two directions: we make the constant factor explicit  in the estimate  $\sC_\Tx(n)\in \bigO\bigl(\MM(n)\bigr),$ and lessen it.
Then, we prove that $\MMul nnn$, $\OMul nn\Dx$, and $\OMul nn\Tx$ are equivalent computational problems, in a sense made clear below.

\subsection{Product in  $ \bK[\x]\langle \Tx  \rangle$ reduces to matrix product: van der Hoeven's algorithm revisited}\label{ssec:OperToMat}

A differential operator~$A$ in $\bK[\x]\langle\Tx\rangle$ can be viewed as a $\bK$-endomorphism of~$\bK[\x]$, mapping a polynomial~$f$ to~$A(f)$.
As such, it is represented, with respect to the canonical basis $(\x^i)_{i \geq 0}$ of $\bK[\x]$, by an (infinite) matrix  denoted $M_\infty^A$.
The submatrix of $M_\infty^A$ consisting of its first $r\geq 1$ rows and $c\geq 1$ columns is denoted $M_{r,c}^A$.

Van der Hoeven's key observation is that an operator~$A$ of bidegree $(d,r)$ is completely determined by the matrix  $M^A:= M_{d+r+1,r+1}^A$.
Writing $A=\sum_{i=0}^d \sum_{j=0}^r a_{i,j} \x^i \Tx^j$ and using the relation $\Tx^j(\x^k) = k^j \x^k$ yields
\[A(\x^k) = \sum_{i,j} a_{i,j} k^j  \x^{i+k} = \x^k \sum_{i=0}^d \tilde A_i(k) \x^i,\]
where the polynomials $\tilde A_i$ are defined as $\tilde A_i(\x) = \sum_{j=0}^r a_{i,j}\x^j$ for all $0\leq i \leq d$.
Thus the matrix $M^A$ has the following rectangular banded form:

\begin{equation} \label{matrixdef}
M^A =
\left[
\begin{array}{cccc}
 \tilde A_0(0) & & &  \\
\tilde A_1(0) & \tilde A_0(1) & & \\
\vdots & \tilde A_1(1) & \ddots  & \\
\tilde A_d(0) &  \vdots  &  \ddots &\tilde A_0(r)  \\
 & \tilde A_d(1)  &  &\tilde A_1(r)  \\
 & & \ddots & \vdots \\
 & & & \tilde A_d(r)\\
 \end{array}
\right].
\end{equation}

The knowledge of~$A$ is equivalent to that of all $d+1$ polynomials~$\tilde A_i$.
Each of the latter having degrees bounded by~$r$, this is also equivalent to the data of the values $\tilde A_i(k)$, for $0 \leq k \leq r$ and $0\leq i \leq d$.
This is true by Lagrange interpolation.
Thus, $A$ is indeed completely determined by the $r+1$ polynomials $A(\x^k)$, 
and also by the matrix~$M^A$.

Now, let  $A,B \in \bK[\x]\langle \Tx  \rangle$ and let~$C$ be~$BA$.
Then $M_\infty^C = M_\infty^B M_\infty^A$.
If $A$, $B$, and~$C$ have bidegrees $(d_A,r_A),(d_B,r_B)$, and $(d_C,r_C)$,  then the previous discussion implies the following ``finite version'' of this matrix equality:
\begin{equation}\label{matrixproduct}
M^C = M^B_{d_C+r_C+1,d_A+r_C+1} M^A_{d_A+r_C+1,r_C+1},
\end{equation}
which is the basis of the algorithm in~\cite{vdHoeven02}, described below.
\begin{figure}[ht]
  \begin{center}
    \fbox{\begin{minipage}{8cm}
  \begin{center}\textsf{Mul}${}_\Tx$($B,A$) \end{center}
      \textbf{Input:} $A,B \in \bK[\x] \langle \Tx \rangle$.\\
     \textbf{Output:} their product $C=BA$.\\[-4.5mm]
        \begin{tabbing}
1. Compute the $\tilde A_i$'s and~$\tilde B_i$'s from $A$ and~$B$, then the\\
\qquad matrices $M^B_{d_C+r_C+1,d_A+r_C+1} $ and $M^A_{d_A+r_C+1,r_C+1}$.\\
2. Compute $M^C$ by Eq.~\eqref{matrixproduct}.\\
3. Compute the $\tilde C_i$'s from~$M^C$, then recover~$C$.
     \end{tabbing}
      \end{minipage}
    }\end{center}
\vskip-10pt
  \caption{Product of differential operators in $\Tx$.}
  \label{fig:JorisAlgo}
\end{figure}

Putting all these considerations together leads to Algorithm \textsf{Mul}${}_\Tx$ in Fig.~\ref{fig:JorisAlgo} and
proves the following proposition.

\begin{proposition}
Algorithm \textsf{Mul}${}_\Tx$ in Fig.~\ref{fig:JorisAlgo} reduces the computation of the product $C=BA$ to the following tasks:
\begin{enumerate}
\item[\emph{(T1)}\/] $d_A+1$ evaluations in degrees\/~$\leq r_A$ on\/ $0,1,\ldots,r_C$;
\item[\emph{(T2)}\/] $d_B+1$ evaluations in degrees\/~$\leq r_B$ on\/ $0,1,\ldots,d_A+r_C$;
\item[\emph{(T3)}\/] $d_C+1$ interpolations in degrees\/~$\leq r_C$ on\/ $0,1,\ldots,r_C$;
\item[\emph{(T4)}\/] an instance of\/ $\MMul{d_C+r_C+1}{d_A+r_C+1}{r_C+1}$.
\end{enumerate}
\end{proposition}

\begin{proof} Eq.~\eqref{matrixdef} shows that Step~1 in Algorithm \textsf{Mul}${}_\Tx$ is performed by the evaluation Tasks (T1--T2).
Similarly, the interpolation Task~(T3) performs Step~3.
Finally, the product in Step~2 is computed by~(T4).
\end{proof}

We stress that the evaluation-interpolation scheme used in Algorithm \textsf{Mul}${}_\Tx$ requires that the interpolation points $0,1,\ldots,r_C$ be mutually distinct.
Thus, this scheme would not have worked over a field of small characteristic, but would have remained valid in large enough characteristic.

In the original article~\cite{vdHoeven02}, Tasks (T1--T3) are performed by matrix multiplications, as explained in the next lemma.

\begin{lemma} \label{evaluationVand}
Let $d,r,s\in\bN$ and let $a_0,\ldots,a_s$ be distinct points in\/~$\bK$.
Evaluating $d+1$  polynomials of degree $r$ on the $a_i$'s reduces to an instance of\/ $\MMul{s+1}{r+1}{d+1}$ plus $\bigO(sr)$~ops.
Interpolating $d+1$~polynomials of degree~$s$ on the~$a_i$'s amounts to an instance of\/ $\MMul{s+1}{s+1}{d+1}$ plus $\bigO(s^2)$~ops.
\end{lemma}

\begin{proof}
The omitted proof is based on grouping multiplications by Vandermonde matrices into a single product.
\end{proof}

Using Lemma~\ref{evaluationVand}, one immediately deduces the cost of van der Hoeven's algorithm  ``\`a la lettre'' (\textsf{vdH${}_\Tx$});
the following enumeration displays only the dominating costs, quadratic estimates like $O(r_C r_A)$ being intentionally neglected:
\begin{enumerate}
\item $\MM(r_C+1,r_A+1, d_A+1)$ for (T1);
\item  $\MM(d_A+r_C+1,r_B+1,d_B+1)$ for (T2);
\item  $\MM(r_C+1,r_C+1,d_C+1)$ for (T3);
\item  $\MM(d_C+r_C+1,d_A+r_C+1,r_C+1)$ for (T4).
\end{enumerate}
Notice that the last step dominates the cost.

For Problem $\OMul nn \Tx$ which is studied in~\cite{vdHoeven02}, applying the estimate~\eqref{eq:naive-abc} leads to the number~$2+3+2\cdot2\cdot2+4\cdot3\cdot2=37$ of $n \times n$ block multiplications given in column \textsf{vdH${}_\Tx$} of Table~\ref{table:MM}.
This estimate is however pessimistic and can be reduced to~20:
Strassen's formula reduces the 8~block products in Task~(T3) to~7;
the band structure of the matrices in Task~(T4) reduces~24 to only 8~products of non-zero blocks.

\paragraph*{A first improvement}
Algorithm \textsf{vdH${}_\Tx$} can be improved
by making use of fast multipoint evaluation and interpolation of Lemma~\ref{cost-results}(b) to perform Steps 1 and~3 of Algorithm {\textsf{Mul${}_\Tx$}} in Fig.~\ref{fig:JorisAlgo}.
This remark will be crucial in our proof of equivalence in~\S\ref{ssec:MatToOper}.
We arrive at the following complexity estimates:
\begin{enumerate}
\item $\bigO\bigl(d_A \, \sM(r_C) \log r_C \bigr)$ for (T1);
\item $\bigO\bigl(d_B \, \sM(d_A+r_C) \log (d_A+r_C) \bigr)$ for (T2);
\item $\bigO\bigl(d_C \, \sM(r_C) \log r_C \bigr)$ for (T3).
\end{enumerate}

Assuming FFT is available for polynomial multiplication, the cumulated  cost of Tasks (T1--T3) drops to
\[\bigOsoft\bigl(d_A r_C + d_B(d_A+r_C) + d_C r_C\bigr) \in  \bigOsoft\bigl( d_Cr_C + d_A d_B\bigr).\]
This cost is nearly optimal, since it is almost linear in the number of non-zero elements of the matrices involved in Eq.~\eqref{matrixproduct}.
In the particular case of problem $\OMul nn \Tx$, we obtain the numbers 24 and~8 of column \textsf{IvdH${}_\Tx$} in Table~\ref{table:MM}.

\subsection{Matrix multiplication reduces to  product in $ \bK[\x]\langle \Tx  \rangle$} \label{ssec:MatToOper}
In summary, the results of the previous section show that $\sC_\Tx(n)\in \bigO\bigl(\MM(n)\bigr)$.
Here we prove the converse statement, by proceeding in two steps.
First, Lemma~\ref{fromDiffMulToT} shows that the multiplication, whose complexity is denoted $\sT(n)$, of two lower-tri\-an\-gu\-lar matrices of size $n\times n$ reduces to the product of two operators of bidegree at most~$(n,n)$ in $(\x,\Tx)$.
Secondly, Lemma~\ref{fromTtoOmega} proves that multiplying two arbitrary matrices amounts to a constant number of products of lower-triangular matrices.

\begin{lemma}\label{fromDiffMulToT}
  $\sT(n) \in \sC_\Tx(n) +  \bigO\bigl(n\,\sM(n) \log n\bigr).$
\end{lemma}

\begin{proof}
Let $L_1,L_2$ be two $(n+1) \times (n+1)$ lower-triangular matrices.
Denote $(t_{i,j})$ and $(s_{i,j})$ their coefficients, with $0\leq i,j \leq n$.
Let $\tilde B_\ell(\x)$ and $\tilde A_\ell(\x)$ be the (unique) polynomials in $\bK[\x]$ of degree at most $n-\ell$ that interpolate the elements of the $\ell$th lower diagonal of $L_1$, resp.\ $L_2$, on the set $\{0,1,\ldots,n-\ell\}$.
Specifically, for $0\leq \ell,j \leq n$ with $\ell + j \leq n$, we have $t_{\ell+j,j} = \tilde B_\ell(j) $ and $ s_{\ell+j,j} = \tilde A_\ell(j)$.
Using fast interpolation, the computation of the polynomials $\tilde B_\ell(\x)$ and $\tilde A_\ell(\x)$, for $0\leq \ell \leq n$, is done in $\bigO\bigl(n \,\sM(n) \log n\bigr)$~ops.
Define $A=\sum_{\ell=0}^n \sum_{j=0}^{n-\ell} a_{\ell,j} \x^\ell \Tx^j$ and $B=\sum_{\ell=0}^n \sum_{j=0}^{n-\ell} b_{\ell,j} \x^\ell \Tx^j$ from the coefficients in
$\tilde A_\ell(\x) = \sum_{j=0}^{n-\ell} a_{\ell,j}\x^j$ and $\tilde B_\ell(\x) = \sum_{j=0}^{n-\ell} b_{\ell,j}\x^j$.
Let $C=BA$.
Then, $L_1$ and~$L_2$ are seen to be top-left blocks of $M^B$ and~$M^A$, and Eq.~\eqref{matrixproduct} with $A$ replaced by~$C$ shows that the top-left $(n+1) \times (n+1)$ submatrix of $M^C$ is the lower-triangular matrix $L_1 L_2$.
This submatrix is computed starting from the coefficients of $C$ using  $\bigO\bigl(n\, \sM(n) \log n\bigr)$~ops., by fast multipoint evaluation.
\end{proof}

\begin{lemma}\label{fromTtoOmega}
 $\MM(n)  \in \bigO\bigl(\sT(n)\bigr).$
\end{lemma}

\begin{proof}
Let $M, N$ be~$n\times n$ matrices. The identity
\[
\begin{bmatrix}
I_n& 0 & 0 \\
M & I_n & 0 \\
0 & N & I_n
\end{bmatrix}^2
=
\begin{bmatrix}
I_n& 0 & 0 \\
2M & I_n & 0 \\
NM & 2N & I_n
\end{bmatrix}
\]
shows that $\MM(\lceil n/3 \rceil) \leq \sT(n) \in \bigO(\sT(n))$ and the conclusion follows from the growth hypotheses on  $\MM$.
\end{proof}

Lemmas \ref{fromDiffMulToT} and~\ref{fromTtoOmega} imply the main result of this section.
\begin{theorem}
There exists a constant $K>0$ such that
\[\MM(n) \leq K\bigl( \sC_\Tx(n) +   n\,\sM(n) \log n\bigr).\]
\end{theorem}

\subsection{Equivalence between product in $ \bK[\x]\langle \Dx  \rangle$ and in $ \bK[\x]\langle \Tx  \rangle$}
\label{EquivDandTheta}

Relax $\bK$ to be a field of arbitrary characteristic. 
Any operator $A=\sum_{i=0}^r \alpha_i \Tx^i$ in $ \bK[\x]\langle \Tx \rangle$ with coefficients $\alpha_i$ of degree at most $d$ can be expressed in the algebra $ \bK[\x]\langle \Dx \rangle$ as $A=\sum_{i=0}^r a_i \Dx^i$, with coefficients $a_i$ of degree at most $d+r$.

As indicated in the proof of \cite[Cor.~2]{BoSc05}, performing the conversion from the representation in~$\Tx$ to the representation in~$\Dx$ amounts to multiplying a Stirling matrix~$S$ of size~$r+1$ by an $(r+1) \times (d+1)$ matrix containing the coefficients of the $\alpha_i$'s.
This matrix product can be decomposed into $d+1$ matrix-vector products of the form $w=Sv$.
The coefficients of the vector~$w$ represent the coefficients of the polynomial $\sum_i v_i \x^i$ in the falling factorial basis $(\x)_k$.
As Lemma~\ref{cost-results}(c) holds for any characteristic, $w$~can be computed using $\bigO\bigl(\sM (r) \log r\bigr)$~ops.
To summarize, the coefficients~$a_i$ can be computed from the~$\alpha_i$'s in $\bigO\bigl(d \, \sM (r) \log r\bigr)$~ops.

Conversely, let $B= \sum_{i=0}^r b_i \Dx^i$ be  in $ \bK[\x]\langle \Dx \rangle$.
It can be written in the algebra $ \bK[\x,\x^{-1}]\langle \Tx  \rangle$ of differential operators in $\Tx$ with Laurent polynomial coefficients as follows: $B=\sum_{i=0}^r \beta_i \Tx^i$.
If the $b_i$'s have degrees bounded by $d$, then the $\beta_i$'s have degrees at most~$d$ and valuation at least~$-r$ in~$\x$.
A discussion similar to above shows that the computation of the coefficients $\beta_i$ from the coefficients $b_i$ amounts to multiplying the inverse of the Stirling matrix by an $(r+1) \times (d+r+1)$ matrix.
This matrix product can be decomposed into $d+r+1$ matrix-vector products by $S^{-1}$;
this amounts to expanding in the monomial basis  $d+r+1$ polynomials of  degree at most~$r$ given in the falling factorial basis.
Thus, the conversion can be done in $\bigO\bigl((d+r) \, \sM (r) \log r\bigr)$~ops.

We encapsulate this discussion into the following result, which proves that $\OMul nn\Dx$ and~$\OMul nn\Tx$ are computationally equivalent, up to $\bigOsoft(n^2)$ terms, in any characteristic.

\begin{theorem}
 There exist a constant $C>0$ such that
\begin{align*}
\sC_\Tx(n) \leq C\,\bigl(\sC_\Dx(n) + n \, \sM(n) \log n\bigr), \\
\sC_\Dx(n) \leq C\,\bigl(\sC_\Tx(n) + n \, \sM(n) \log n\bigr) ,
\end{align*}
% To avoid a widow on the next page.
\vskip-1pt
over fields of any characteristic.
\end{theorem}

\begin{proof}
Let $A_1,A_2$ be of bidegree~$(n,n)$ in $\bK[\x]\langle \Tx \rangle$.
By the previous discussion, converting them into $\bK[\x]\langle \Dx \rangle$ has cost $\bigO\bigl(n \, \sM(n) \log n\bigr)$.
Both $A_1,A_2$ have bidegrees at most $(2n,n)$ in $(\x,\Dx)$ and can thus be multiplied using $\sC_\Dx(2n) \in \bigO\bigl(\sC_\Dx(n)\bigr)$~ops.\ 
Converting the result back into $\bK[\x]\langle \Tx \rangle$ costs $\bigO\bigl(n \, \sM(n) \log n\bigr)$ ops.
This proves the first inequality.

Let now $B_1$ and~$B_2$ be of bidegree $(n,n)$ in $\bK[\x]\langle \Dx \rangle$.
Their conversion in $\bK[\x,\x^{-1}]\langle \Tx \rangle$ can be performed using $\bigO\bigl(n \, \sM(n) \log n\bigr)$~ops.\ and produces  two operators $C_1$ and $C_2$ in $\bK[\x]\langle \Tx \rangle$, of bidegrees at most $(2n,n)$ in $(\x,\Tx)$ such that $B_1 =\x^{-n} C_1(\x,\Tx)$ and $B_2 =\x^{-n} C_2(\x,\Tx)$. 
Using the commutation rule $C_2(\x,\Tx)  \x^{-n} = \x^{-n} C_2(\x,\Tx-n)$, we deduce the equality $B_2 B_1 =  \x^{-2n} C_2 (\x,\Tx-n) C_1(\x,\Tx)$.
Writing $C_2(\x,\Tx) = \sum_{j=0}^n \x^j c'_j(\Tx)$ shows that computing the coefficients of $C_2 (\x,\Tx-n)$ amounts to $n+1$ polynomial shifts in $\bK[\Tx]$ in degree at most $n$.
Each of these shifts can be computed in $\bigO\bigl(\sM(n) \log n\bigr)$~ops., using Lemma~\ref{cost-results}(a).
Conversion of $B_2B_1$ back into $\bK[\x]\langle \Dx \rangle$ has the same cost.
\end{proof}

\section{Better constants in $\bK[\x]\langle\Dx\rangle$}
\label{sec:better-constants}
In~\S\ref{ssec:vdH_D}, we revisit van der Hoeven's algorithm for $\OMul nn\Dx$ and exhibit the constant factor in its $\bigO\bigl(\MM(n)\bigr)$ cost. Then, we propose in~\S\ref{sec:MulWeyl} a new algorithm with a better constant.

\subsection{Multiplication in $\bK[\x]\langle\Dx\rangle$: van der Hoeven's algorithm revisited}
\label{ssec:vdH_D}

Van der Hoeven's algorithm for computing products in $\bK[\x]\langle \Dx \rangle$ is based on the fact that his algorithm for products in $\bK[\x]\langle \Tx  \rangle$ can be adapted to operators with Laurent polynomials coefficients.
Indeed, to any  $L \in \bK[\x,\x^{-1}]\langle \Tx  \rangle$ of the form $L=\sum_{i=-v}^d \sum_{j=0}^r \ell_{i,j}\x^i \Tx^j$ is associated an infinite matrix representing the $\bK$-linear map of multiplication by~$L$ from~$\bK[\x]$ to~$\x^{-v}\bK[\x]$.
Its $(v+d+r+1) \times (r+1)$-submatrix~$M^L_{0,r}$ (defined shortly) is banded and it uniquely determines the operator~$L$, as in the case of polynomial coefficients.

To be precise, for  two integers $\alpha\leq\beta$ we denote by $M^L_{\alpha,\beta}$ the $(v+d+\beta-\alpha +1) \times (\beta - \alpha + 1)$ matrix whose $(\gamma-\alpha+1)$-th column,  for $\alpha\leq\gamma \leq \beta$,  contains the coefficients of $L(\x^\gamma)$ on $\x^{-v+\alpha},\ldots,\x^{d+\beta}$.
The matrix $M^L_{\alpha,\beta}$ has a banded form and contains on its diagonals the evaluations on the points $\alpha,\ldots, \beta$ of the polynomials $\tilde{L}_{-v},\ldots,\tilde{L}_d$ defined by $\tilde{L}_i(\x) = \sum_{j=0}^r \ell_{i,j}\x^j$ for all $-v \leq i \leq d$.

Let $A,B$ have valuations $-v_A,-v_B$ and degrees $d_A,d_B$ with respect to~$\x$, and degrees $r_A,r_B$ in~$\Tx$.
If $C=BA$  in $\bK[\x,\x^{-1}]\langle \Tx  \rangle$, then the  following equality, analogous to Eq.~\eqref{matrixproduct}, holds:
\begin{equation} \label{matrixproduct-Laurent}
M^C_{0,r_C} = M^B_{-v_A,d_A+r_C} M^A_{0,r_C}.
\end{equation}

Likewise, the product of operators in $\bK[\x,\x^{-1}]\langle \Tx  \rangle$ reduces to some evaluation and interpolation tasks (in order to convert between operators and matrices) and to the main matrix-multiplication task~\eqref{matrixproduct-Laurent}, which is an instance of $\MMul{v_C+d_C+r_C+1}{v_A+d_A+r_C+1}{r_C+1}$.

The algorithm for multiplication in $\bK[\x]\langle \Dx  \rangle$ based on multiplication in $\bK[\x,\x^{-1}]\langle \Tx  \rangle$ is described in Fig.~\ref{fig:JorisAlgo-Laurent} below.

\begin{figure}[ht]
  \begin{center}
    \fbox{\begin{minipage}{8cm}
  \begin{center}\textsf{Mul}${}_\Dx$($B,A$) \end{center}
      \textbf{Input:} $A,B \in \bK[\x] \langle \Dx \rangle$.\\
     \textbf{Output:} their product $C=BA$.\\[-4.5mm]
        \begin{tabbing}
1. Convert $A,B$ in $\bK[\x,\x^{-1}]\langle \Tx  \rangle$. \\
2. Compute the product $C=BA$ in $\bK[\x,\x^{-1}]\langle \Tx  \rangle$:\\
\qquad 2.1 From $A$ and $B$, compute the matrices   \\  \qquad \qquad $ M^B_{-v_A,d_A+r_C}$ and $M^A_{0,r_C}.$ \\
\qquad 2.2 Compute the matrix $M^C_{0,r_C}$ using Eq.~\eqref{matrixproduct-Laurent}. \\
\qquad 2.3 Recover $C$ {}from $M^C_{0,r_C}.$ \\
3. Convert $C$ in $\bK[\x]\langle \Dx  \rangle$ and return it.
     \end{tabbing}
      \end{minipage}
    }\end{center}
\vskip-10pt
  \caption{Product of differential operators in $\Dx$.}
  \label{fig:JorisAlgo-Laurent}
\end{figure}

In what follows, we treat in more detail  the main case of interest, $\OMul nn\Dx$, as solved by Algorithm \textsf{Mul}$_\Dx$ in Fig.~\ref{fig:JorisAlgo-Laurent}. 
Van der Hoeven suggests to perform Steps 1 and~3 using matrix multiplications by Stirling matrices and their inverses~\cite[\S5.1, Eqs.~(12--13)]{vdHoeven02} and Steps 2.1 and~2.3 using matrix multiplications by Vandermonde matrices and their inverses~\cite[\S2 and~\S4]{vdHoeven02}.
The elements of all the needed Stirling and Vandermonde matrices (and their inverses) can be computed using $\bigO(n^2)$~ops.
A careful inspection of the matrix sizes involved in Algorithm \textsf{Mul}$_\Dx$ shows that: 
 \begin{enumerate}
\item Step~1 reduces to 2 instances of $\MMul{2n+1}{n+1}{n+1}$;
\item Step~3 reduces to an instance of $\MMul{4n+1}{2n+1}{2n+1}$;
\item Step~2.1 reduces to an instance of $\MMul{2n+1}{n+1}{4n+1}$ and an instance of $\MMul{2n+1}{n+1}{2n+1}$;
\item Step~2.3 reduces to an instance of $\MMul{4n+1}{2n+1}{2n+1}$;
\item Step~2.2 reduces to an instance of $\MMul{6n+1}{4n+1}{2n+1}$.
 \end{enumerate}
This variant of the algorithm is what we call \textsf{vdH}.
Using again the estimate~\eqref{eq:naive-abc} yields the constant~96 in Table~\ref{table:MM}.

\paragraph*{Several Improvements} A first improvement on \textsf{vdH} 
is to use fast multipoint evaluation and interpolation for Steps 2.1 and~2.3.
A second improvement concerns conversions back and forth between operators in $\bK[\x,\x^{-1}]\langle \Dx  \rangle$ and in $\bK[\x,\x^{-1}]\langle \Tx  \rangle$ (Steps 1 and~3).
Instead of using matrix products by Stirling matrices and their inverses, one can apply Lemma~\ref{cost-results}(c),
as explained in~\S\ref{EquivDandTheta}.
Both improvements in conjunction with FFT lessen the cost of Steps 1, 2.1, 2.3, and 3 to a negligible $\bigOsoft(n^2)$.
We call this improved algorithm \textsf{IvdH}.
Using~\eqref{eq:naive-abc} yields the constant~48 in column \textsf{IvdH} in Table~\ref{table:MM}.
The constants 47 and~12 on the last row of the table are more technical and will be proved in~\cite{LongVersion}.
They rely on observing that the output of \textsf{IvdH} requires partial calculation of~\eqref{matrixproduct-Laurent}, reducing to an instance of $\MMul{4n+1}{3n+1}{2n+1}$.

\subsection{A new, direct evaluation-interpolation algorithm}
\label{sec:MulWeyl}

Let $A$ and~$B$ be in  $\bK[\x]\langle \Dx \rangle$ with respective bidegrees $(d_A,r_A)$ and~$(d_B,r_B)$.
We give here an evaluation-inter\-pol\-ation algorithm for computing $C = BA$ which essentially reduces to $\MMul{d_C+1}{d_A + r_C +1}{r_C+1}$ for those bidegrees.

To achieve this, we interpret again a differential operator~$P$ in $\bK[\x]\langle\Dx\rangle$ as a $\bK$-endomorphism of~$\bK[\x]$, and represent it in the canonical basis $(\x^i)_{i \geq 0}$ by an (infinite) matrix  denoted $\tM_\infty^P$.
The submatrix of $\tM_\infty^P$ consisting of its first $r+1\geq 1$ rows and $c+1\geq 1$ columns is denoted~$\tM_{r,c}^P$.

Then, much like Algorithm \textsf{Mul}${}_\Tx$ in~\S\ref{ssec:OperToMat}, our new algorithm \textsf{MulWeyl} in Fig.~\ref{fig:AlgoMulWeyl} relies on the key observation that an operator $P \in \bK[\x]\langle \Dx \rangle$ of bidegree $(d,r)$ is uniquely determined by the submatrix $\tM_{d,r}^P$ of~$\tM^P_\infty$.
This key fact is proved in Theorem~\ref{interpol} below.
The principle of the algorithm is given in Fig.~\ref{fig:EvalInterp}, where evaluation and interpolation are performed by truncated-series products.
\begin{figure}[h]
\begin{center}
\includegraphics[scale=0.8]{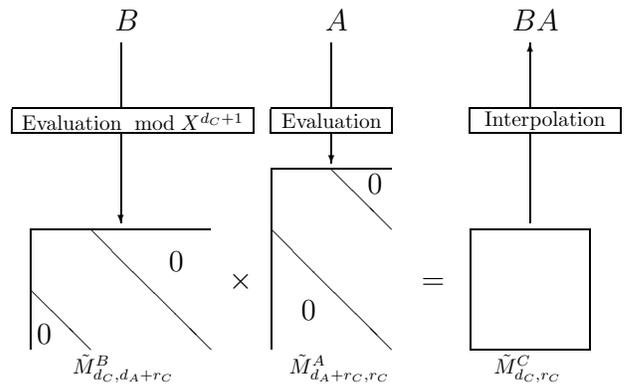}
\end{center}
\caption{\label{fig:EvalInterp} Evaluation-Interpolation w.r.t.~$\Dx$.}
\end{figure}
In the case of~$\OMul nn\Dx$, the corresponding matrices become $\tM^B_{2n,3n}$, $\tM^A_{3n,2n}$, and $\tM^C_{2n,2n}$.

\begin{figure}[ht]
  \begin{center}
    \fbox{\begin{minipage}{8cm}
  \begin{center}\textsf{MulWeyl}($B,A$) \end{center}
      \textbf{Input:} $A,B \in \bK[\x] \langle \Dx \rangle$.\\
     \textbf{Output:} their product $C=BA$.\\[-4.5mm]
\begin{tabbing}
1. Construct the matrices $\tM^A_{d_A+r_C,r_C}$ and $\tM^B_{d_C,d_A+r_C}$. \\
2. Compute the product $\tM^B_{d_C,d_A+r_C}\tM^A_{d_A+r_C,r_C}$. \\
3. Recover $C$ from the product in Step~2.
\end{tabbing}
\end{minipage}
    }\end{center}
\vskip-10pt
  \caption{Product of differential operators in $\Dx$.}
  \label{fig:AlgoMulWeyl}
\end{figure}

\begin{theorem}\label{theo:cost-MulWeyl}
Algorithm \textsf{MulWeyl} is correct and uses\/\\
$\MM(d_C+1, d_A + r_C +1, r_C+1) + \bigOsoft\bigl((d_C+r_C)^2\bigr)$~ops.
\end{theorem}

\begin{proof}
By the definition of the matrix~$\tM_{r,c}^P$, the matrices constructed in Step~1 are associated to the linear map which sends $f \in \bK[\x]_{ \leq r_C}$ to~$A(f)$ in $\bK[\x]_{\leq d_A+r_C}$ and to the linear map which sends $f \in \bK[\x]_{\leq d_A+r_C}$ to $B(f) \bmod \x^{d_C + 1}$ in $\bK[\x]_{\leq d_C}$.
Therefore, the product at Step~2 delivers the $(BA)(\x^i) \bmod \x^{d_C+1}, \ 0 \leq i \leq r_C$.
The $(d_C+1) \times (r_C+1)$ matrix computed is thus equal to $\tM^C_{d_C,r_C}$.
This is summarized in the identity $\tM^B_{d_C,d_A+r_C} \tM^A_{d_A+r_C,r_C} = \tM^C_{d_C,r_C}$, in which the structure of zeros is given in Fig.~\ref{fig:EvalInterp}.
The interpolation of Step~3 relies on Theorem~\ref{interpol} below, which shows that $C = BA$ is fully and uniquely determined by $\tM^C_{d_C,r_C}$.
This terminates the correctness proof.
The claimed complexity derives immediately from Propositions \ref{prop:Eval} and~\ref{prop:Interpol} that are proved in the next subsections.
\end{proof}

\subsubsection{Interpolation theorem}

We now state the main interpolation result, which we prove after recalling a useful filtration on~$W=\bK[\x]\langle\Dx\rangle$.

\begin{theorem}\label{interpol}
For $d,r\in\bN$, let $W_{d,r}$ denote its $\bK$-subspace
\[W_{d,r} = \lbrace\, P \in W \; : \;  \deg_{\x}(P) \leq d, \ \deg_\Dx(P) \leq r \,\rbrace.\]
Then, an isomorphism is given by the $\bK$-linear map
\[\begin{array}{lccc}
\operatorname{EvOp}_{d,r} \, : \, & W_{d,r} & \rightarrow & \bK^{(d+1)\times(r+1)}\\
                                                                           & P             & \mapsto  & \tM^P_{d,r}
\end{array}.\]
\end{theorem}

In order to prove Theorem~\ref{interpol}, we use the filtration~on $W$ defined by the weights $1$ on~$\x$ and $-1$ on~$\Dx$.
The decomposition into homogeneous components of any $P\in W_{d,r}$ only involves weights between $-r$ and~$d$.
It actually admits a special form, to be exploited later, which is described now.

\begin{lemma}\label{lem:b-hom-parts}
The homogeneous decomposition of $P\in W_{d,r}$ is
\[P = \sum_{i=1}^r \ell_{-i}(\x\Dx) \Dx^i + \sum_{i=0}^d \x^i \ell_i(\x\Dx),\]
where the $\ell_i$'s and $\ell_{-i}$'s are polynomials of degree at most $\mu_i:=\min(d-i,r)$ and $\mu_{-i}:=\min(r-i,d)$, respectively.
\end{lemma}

\begin{proof}
Let $P$ be $\sum_{i,j} p_{i,j} \x^j \Dx^i$.
Then $P$ decomposes as the sum of $\sum_{i>j} p_{i,j} (\x^j \Dx^j) \Dx^{i-j}$ and of $\sum_{i \leq j} \x^{j-i} p_{i,j} (\x^i \Dx^i)$.
Here $p_{i,j}$ is zero if $i>r$ or if $j>d$, therefore $P$ is equal to  
\begin{equation}\label{decomp}
	\sum_{s=1}^r \left(\sum_{j=0}^{\mu_{-s}} p_{j+s,j} \x^j \Dx^j \right) \Dx^s + \sum_{t=0}^d \x^t \left(\sum_{i=0}^{\mu_t} p_{i,t+i} \x^i \Dx^i \right).
\end{equation}	
Since any $\x^i \Dx^i$ can be written as a polynomial of degree~$i$ in $\x\Dx$, the conclusion follows by expressing each parenthesis in~\eqref{decomp} as a polynomial in $\x\Dx$.
\end{proof}

\begin{proof}[of Th.~\ref{interpol}]
Since $\dim_{\bK}W_{d,r}$ is $\dim_{\bK}\bK^{(d+1)\times(r+1)}$, it suffices to show that $\operatorname{EvOp}_{d,r}$ is injective.
Let $P$ in~$W_{d,r}$ be such that $\tM^P_{d,r}=0$, or equivalently $P(\x^k) \bmod \x^{d+1} = 0$ for all $0\leq k \leq r$.
The decomposition of~$P\in W_{d,r}$ in Lemma~\ref{lem:b-hom-parts} enables one to evaluate it easily at~$\x^k$ for~$k\leq r$:
\begin{equation}\label{evaluation}
P\bigl(\x^k\bigr) = \sum_{i=1}^k \frac{k!}{(k-i)!} \ell_{-i}(k-i) \x^{k-i} + \sum_{i=0}^d \ell_{i}(k) \x^{k+i}.
\end{equation}
Since $P(\x^k)\bmod\x^{d+1}=0$ for $k\leq r$, Eq.~\eqref{evaluation} implies:
\begin{itemize}
	\item $\ell_i(k)=0$ if $0\leq i\leq d$, $0\leq k\leq r$, and $k+i\leq d$,
	\item $\ell_{-i}(k-i)=0$ if $1\leq i\leq k$, $0\leq k\leq r$, and $k-i\leq d$.
\end{itemize}
These equalities show that $\ell_i(0),\dots,\ell_i\bigl(\min(d-i,r)\bigr)$ are zero for $0\leq i\leq d$ and that $\ell_{-i}(0),\dots,\ell_{-i}\bigl(\min(r-i,d)\bigr)$ are zero for $1\leq i\leq r$.
Finally, Lagrange interpolation and the degree bounds in Lemma~\ref{lem:b-hom-parts} imply that all the polynomials $\ell_i$ and $\ell_{-i}$ are identically zero.
Thus, $P$~is~0.
\end{proof}

A direct use of the ideas of this subsection would now end the proof of Theorem~\ref{theo:cost-MulWeyl};
the corresponding algorithm would first compute the polynomials $\ell_i$ and $\ell_{-i}$, before evaluating them on $0,1,\ldots$.
By the following next two subsections, we shall propose a better solution, avoiding a logarithmic factor and hiding a smaller constant in the $\bigOsoft({\cdot})$ term.

\subsubsection{Evaluation step}

Here we focus on Step~1 of Algorithm \textsf{MulWeyl}, which is an instance of the task of computing the matrix $\tM^P_{m,n}$ for given $P = \sum_{i=0}^r \sum_{j=0}^d p_{i,j} \x^j \Dx^i$  in $W$ and integers $m \geq d,n \geq r$.
The announced better approach makes use of Algorithm \textsf{Eval} in Fig.~\ref{fig:AlgoEval}, which is based on the following observation:
Let $0 \leq k \leq n$.
Then we have the identities
\begin{align*}
P \left( \x^k \right) & = \sum_{i=0}^{\min(r,k)} \sum_{j=0}^{d} p_{i,j} \frac{k!}{(k-i)!} \x^{k+j-i} \\
  & = k! \, \x^k \Biggl( \sum_{\ell = - \min(r,k)}^d  \biggl(\sum_{i=\max(0,-\ell)}^{\min(r,d-\ell,k)} \frac{p_{i,i+\ell}}{(k-i)!}\biggr) \,\x^\ell\Biggr).
\end{align*}
Therefore, for $-r \leq \ell \leq d$ and  $0 \leq k \leq n$, the coefficient~$(S_\ell)_k$ of~$\x^k$ in the polynomial product
\begin{equation} \label{eq:convolution} S_\ell =  \biggl(\sum_{i=\max(0,-\ell)}^{\min(r,d-\ell)} p_{i,i+\ell}\x^i\biggr) \, \biggl(\sum_{j=0}^n \frac{\x^j}{j!} \biggr)\end{equation}
gives the coefficient of $\x^\ell$ in $\bigl(k!\,\x^k\bigr)^{-1}P(\x^k)$.
Thus the coefficients $(S_\ell)_k$ for $\max(0,-\ell) \leq k \leq \min(m-\ell,n)$ of $S_\ell$ are, up to factorials, the coefficients on a certain diagonal of the matrix $\tM^P_{m,n}$, the other diagonals of $\tM^P_{m,n}$ being zero. 

\begin{figure}[ht]
  \begin{center}
    \fbox{\begin{minipage}{8cm}
  \begin{center}\textsf{Eval}($P$) \end{center}
    \textbf{Input:}    $P \in W_{d,r}, \, m \geq d,\,n\geq r$.\\
     \textbf{Output:} $\tM^P_{m,n}$.\\[-4.5mm]
\begin{tabbing}
1. For each $-r \leq \ell \leq d$, compute $S_\ell \bmod \x^{\min(m-\ell,n)+1}$ \\
\qquad by using Eq.~\eqref{eq:convolution}. \\
2. Initialize $M$ to be an $(m+1) \times (n+1)$ zero matrix. \\
3. For $-r \leq \ell \leq d$ and $\max(0,-\ell) \leq k \leq \min(m-\ell,n)$,\\
\qquad$M_{k+\ell,k} := k! \, (S_\ell)_k$.
\end{tabbing}
\end{minipage}
    }\end{center}
\vskip-10pt
  \caption{Evaluation in $\bK[\x]\langle \Dx \rangle$.}
  \label{fig:AlgoEval}
\end{figure}

\begin{proposition}\label{prop:Eval}
Algorithm \textsf{Eval} computes $\tM^P_{m,n}$ in\/ $\sM(mn)+\bigO(mn)$~ops.
\end{proposition}

\begin{proof}
The series $\exp(\x) \bmod \x^{n+1}$ and the factorials 1, \dots, $n!$ are computed by recurrence relations in $\bigO(n)$~ops.
The computation of $S_\ell$ can be done in~$\sM(s_\ell)$ for the size~$s_\ell$ of the corresponding diagonal of $\tM^P_{m,n}$.
Summing over~$\ell$ and appealing to properties of~$\sM$ leads to $\sum_\ell\sM(s_\ell)\leq\sM\bigl(\sum_\ell s_\ell\bigr)$ $\leq\sM(mn)+\bigO(mn)$, then to the announced complexity.
\end{proof}

\subsubsection{Interpolation step} 

Given a $(d+1)\times (r+1)$ matrix $M$, Step~3 of Algorithm \textsf{MulWeyl} computes the only operator $P \in W_{d,r}$ satisfying $\tM^{P}_{d,r}=M$.
This is done by inverting Eq.~\eqref{eq:convolution}.
The resulting algorithm is described in Fig.~\ref{fig:AlgoInterpol}.
A similar analysis to that of algorithm \textsf{Eval} leads to the estimate in Proposition~\ref{prop:Interpol}.
\begin{figure}[ht]
  \begin{center}
    \fbox{\begin{minipage}{8cm}
  \begin{center}\textsf{Interpol}($M$) \end{center}
      \textbf{Input:} $M \in \bK^{(d+1)\times (r+1)}$.\\
     \textbf{Output:} $P \in W_{d,r}$ such that $\tM^P_{d,r} = M$.\\[-4.5mm]
\begin{tabbing}
1. Divide the $k$th column of $M$ by $k!$.\\
2. For each $-r \leq \ell \leq d$, compute the product $T_\ell={}$\\
$\biggl(\sum\limits_{k=\max(0,-\ell)}^{\min(d-\ell,r)} M_{k+\ell,k} \x^k\biggr) \exp(-\x) \bmod \x^{\min(d-\ell,r)+1}$.\\
3. Return $\displaystyle \sum\limits_{i=0}^r \; \sum\limits_{\ell = - \min(i,r)}^{\min(d-i,r)} (T_\ell)_i \x^{\ell+i} \Dx^i$.
\end{tabbing}
\end{minipage}
    }\end{center}
\vskip-10pt
  \caption{Interpolation in $\bK[\x]\langle \Dx \rangle$.}
  \label{fig:AlgoInterpol}
\end{figure}

\begin{proposition}\label{prop:Interpol}
\textsf{Interpol} computes $P$ in $\sM(dr)+\bigO(dr)$~ops.
\end{proposition}

\subsection{Comparison of algorithms for $\OMul nn\Dx$}

Algorithms \textsf{Mul}$_\Dx$ in Fig.~\ref{fig:JorisAlgo-Laurent} and \textsf{MulWeyl} in Fig.~\ref{fig:AlgoMulWeyl} follow the same scheme:
construction of evaluation matrices associated to $A$ and~$B$; product of these matrices; reconstruction of~$C$ by interpolation from it.
But they  differ in the way to do this, and \textsf{MulWeyl} can be viewed as an improvement on~\textsf{Mul}$_\Dx$:
the matrices computed by \textsf{MulWeyl} are submatrices of $M^B$ and~$M^A$ in Algorithm \textsf{Mul}$_\Dx$, as will be proved in~\cite{LongVersion}.
Taking accurate sizes into account for $\OMul nn\Dx$, the dominant matrix-product problem drops from $\MMul {6n+1}{4n+1}{2n+1}$ to $\MMul {2n+1}{3n+1}{2n+1}$.
Estimate~\eqref{eq:naive-abc} yields the number~12 in the last column of Table~\ref{table:MM}.
Observing that the product at Step~2 of \textsf{MulWeyl} reduces to one instance of $\MMul {2n+1} {2n+1} {2n+1}$ and one of $\MMul{n}{n}{n}$, and appealing to Strassen's formula again, we obtain $7+1=8$ block products, as given on the last row of Table~\ref{table:MM}.

\section{Product in characteristic ${}>0$}
\label{sec:positive-char}

As already pointed out, the evaluation-interpolation algorithms of  Sections~\ref{sec:equiv} and~\ref{sec:better-constants} remain valid when the characteristic $p$ of $\bK$ is positive and sufficiently large, but they fail to work in small characteristic.
For instance, \textsf{MulWeyl} solves Problem~$\OMul nn\Dx$ for characteristic~$p > 3n$.

In this section, we provide an algorithm of different nature which proves that, in characteristic~$p$, the product of two operators of bidegree $(n,n)$ either in $\bK[\x]\langle\Tx\rangle$ or in $\bK[\x]\langle\Dx\rangle$ can be computed in $\bigOsoft(p n^2)$~ops.
For small~$p$, this result is nearly optimal, since it is softly linear in the output size.

Up to $\bigOsoft(n^2)$ additional~ops., multiplication in $\bK[\x]\langle\Dx  \rangle$ can be reduced to multiplication in $\bK[\x]\langle\Tx\rangle$, as explained in~\S\ref{EquivDandTheta}.
Thus, we focus on Problem~$\OMul nn\Tx$.

Our algorithm \textsf{Mul}$_{\Tx,p}$ for multiplication in $ \bK[\x]\langle \Tx  \rangle$ is given in Fig.~\ref{fig:pAlgo}.
It is based on the key fact that $\Tx$ and~$\x^p$ commute in characteristic $p$.
This is used in Step~2, which reduces the product in $\bK[\x]\langle \Tx \rangle$ to several products in the commutative polynomial ring $\bK[\x^p,\Tx]$.

\begin{figure}[ht]
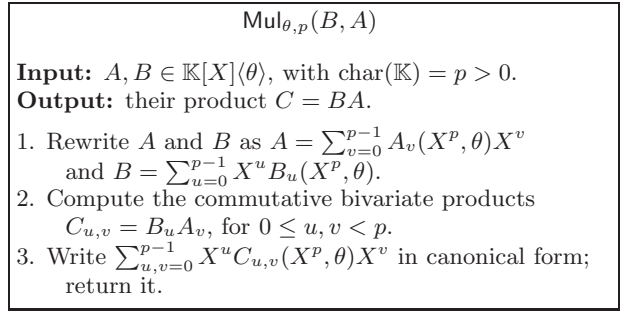

  \begin{center}
    \fbox{\begin{minipage}{7.8cm}
  \begin{center}\textsf{Mul}$_{\Tx,p}$($B,A$) \end{center}
      \textbf{Input:} $A,B \in \bK[\x] \langle \Tx \rangle$, with $\text{char}(\bK) = p> 0$.\\
     \textbf{Output:} their product $C=BA$.\\[-4.5mm]
        \begin{tabbing}
1. Rewrite $A$ and $B$ as $A= \sum_{v=0}^{p-1} A_v(\x^p,\Tx) \x^v$ \\
\qquad and $B= \sum_{u=0}^{p-1} \x^u B_u(\x^p,\Tx)$.\\
2. Compute the commutative bivariate products \\ \qquad $C_{u,v}=B_u A_v$, for $0\leq u,v <p$.\\
3. Write $\sum_{u,v=0}^{p-1} \x^u C_{u,v} (\x^p,\Tx) \x^v$ in canonical form; \\
\qquad return it.
     \end{tabbing}
      \end{minipage}
    }\end{center}
\vskip-10pt
  \caption{Product of differential operators in $\Tx$ over a field of positive characteristic.}
  \label{fig:pAlgo}
\end{figure}
We now describe proper algorithmic choices that perform each step of \textsf{Mul}$_{\Tx,p}$ in nearly optimal complexity.

Step~1 first rewrites $A$ as $\sum_{v=0}^{p-1} \x^v \tilde{A}_v(\x^p,\Tx)$  and $B$ as $\sum_{u=0}^{p-1} \x^u B_u(\x^p,\Tx)$, where 
$B_u,\tilde{A}_v,\, 0 \leq u,v \leq p-1$ are polynomials in $\bK[\x^p,\Tx]$ of bidegree at most $(\lfloor n/p \rfloor, n)$; this costs no ops.
The commutation $\Tx^j\x^v=\x^v(\Tx+v)^j$ then enables one to rewrite~$A$ as $\sum_{v=0}^{p-1} A_v(\x^p,\Tx) \x^v$, where $A_v(\x^p,\Tx)$ is $\tilde{A}_v(\x^p,\Tx-v)$. 
Thus, each $A_v$ is obtained by computing $\lfloor n/p \rfloor +1$ shifts of polynomials of degree at most $n$.
By Lemma~\ref{cost-results}(a), this results in $\bigO\bigl(n \, \sM(n) \log n\bigr)$ ops.\ for Step~1.

Each product in Step~2 involves polynomials in $\bK[\x^p,\Tx]$ of bidegree at most $(\lfloor n/p \rfloor,n)$.
Thus using Lemma~\ref{cost-results}(d), Step~2 is performed in
$\bigO\bigl(p^2 \, \sM(n^2/p)\bigr) \subseteq \bigO\bigl(p \, \sM(n^2)\bigr)$~ops.
Note that $C_{u,v}(\x,\y)$ has bidegree at most $(2\lfloor n/p \rfloor, 2n)$.

To perform Step~3, each $C_{u,v}(\x^p,\Tx) \x^v$ is first rewritten as
$\x^v \tilde{C}_{u,v}(\x^p,\Tx)$ by computing $2\,\lfloor n/p \rfloor +1$
shifts of polynomials of degree at most $2n$. This can be done in $\bigO\bigl(pn\, \sM(n) \log n\bigr)$ ops.
{}Finally,  $\bigO(p n^2)$~ops. are sufficient to put $C = \sum_{u=0}^{p-1} \x^u \sum_{v=0}^{p-1} \x^v \tilde{C}_{u,v}(\x^p,\Tx)$ in canonical form.

Summarizing, we have just proved:
\begin{theorem}
Let\/ $\bK$ be a field of characteristic $p$ and let $D$ be one of the operators $\Dx, \Tx$.
Then, two operators of bidegree $(n,n)$ in $\bK[\x]\langle D\rangle$ can be multiplied in $\bigO\bigl(p \, \sM(n^2) + pn \, \sM(n) \log n\bigr)$~ops., thus in $\bigOsoft(pn^2)$~ops.\ when FFT is used.
\end{theorem}

\section{Experiments}\label{sec:Experiments}

Table~\ref{table:exp} provides timings of calculations in \textsf{magma} by implementations of several algorithms and algorithmic variants.
Each row corresponds to calculations on the same pair of randomly generated operators in bidegree $(n,n)$, for $n=10\cdot2^k$.
Coefficients are taken randomly from~$\bZ/p\bZ$ when~$p>0$, the prime used being
$p_1=65521$ (largest prime to fit on 16~bits) and $p_2=4294967291$ (largest prime to fit on 32~bits).
When~$p=0$, computations are performed over~$\bQ$, with random integer input coefficients on 16~bits.

\begin{table}[ht]
\begin{small}
\begin{center}
\setlength{\tabcolsep}{2.25pt}
\begin{tabular}{rr|rrr|rrr|rrrrr}
$p$ & $k$ & S & B & BZ & vdH & Iter & Tak & Rec & Int & BZI & vdHI \\
\hline
$p_1$ & 3 & 0.25 & 0.26 & 0.25 & 0.39 & 0.32 & 1.23 & 0.01 & 0.64 & 5.22 & 59.8 \\
$p_1$ & 4 & 0.95 & 0.97 & 0.95 & 1.68 & 4.13 & 12.09 & 0.03 & 4.37 & 35.0 & 418 \\
$p_1$ & 5 & 4.08 & 4.11 & 4.34 & 8.10 & 37.2 & 123 & 0.20 & 30.2 & 240 & 2793 \\
$p_1$ & 6 & 21.4 & 21.1 & 22.2 & 45.1 & 397 & 1407 & 1.56 & 209 & 1692 & $\infty$ \\
$p_1$ & 7 & 107 & 105 & 104 & 275 & $\infty$ & $\infty$ & 13.3 & 1507 & $\infty$ & $\infty$ \\
\hline
$p_2$ & 3 & 0.50 & 0.63 & 0.62 & 1.08 & 2.25 & 5.61 & 0.08 & 1.10 & 8.00 & 82.2 \\
$p_2$ & 4 & 2.24 & 2.66 & 2.68 & 4.52 & 19.07 & 67.73 & 0.35 & 9.22 & 58.2 & 602 \\
$p_2$ & 5 & 12.2 & 14.5 & 14.1 & 24.4 & 187 & 926 & 1.63 & 75.6 & 420 & $\infty$ \\
$p_2$ & 6 & 88.1 & 111 & 114 & 172 & 2604 & $\infty$ & 9.40 & 770 & 3146 & $\infty$ \\
$p_2$ & 7 & 1961 & 2452 & 2633 & $\infty$ & $\infty$ & $\infty$ & 59.1 & $\infty$ & $\infty$ & $\infty$ \\
\hline
0 & 3 & 9.93 & 12.0 & 11.3 & 28.4 & 6.99 & 24.3 & 0.07 & 0.93 & 16.9 & 309 \\
0 & 4 & 128 & 164 & 164 & 498 & 118 & 725 & 0.27 & 6.89 & 204 & $\infty$ \\
0 & 5 & 2164 & 2737 & 2725 & $\infty$ & 2492 & $\infty$ & 4.37 & 51.4 & 3172 & $\infty$
\end{tabular}
\caption{\label{table:exp}Timings on input of bidegree~$(10\cdot2^k,10\cdot2^k)$.}
\end{center}
\end{small}
\end{table}
\vskip-7pt

The calculations were performed on a Power Mac G5 with two CPUs at 2.7\,GHz, 512\,kB of L2 Cache per CPU, 2.5\,GB of memory, and a bus of speed 1.35\,GHz.
The system used was Mac OS X 10.4.10, running Magma V2.13-15.
Computations killed after one hour are marked~$\infty$.

We provide several variants of our algorithm (S, B, and BZ), as well as various others:
{\bf S:}
direct call to \textsf{magma}'s matrix multiplication in order to compute
$\tM^B_{2n,3n} \tM^A_{3n,2n}$;
{\bf B and BZ:}
block decomposition into $n\times n$ matrices before calling \textsf{magma}'s  matrix multiplication on, respectively, 11~block products (using Strassen's algorithm) and by 8~block products (taking the nullity of 2~blocks into account as well);
{\bf vdH:}
Van der Hoeven's algorithm, as described in~\cite{vdHoeven02}, and optimized as much as possible as the implementation~S above;
{\bf Iter and Tak:}
iterative formulas \eqref{eq:iterative} and \eqref{eq:takayama};
{\bf Rec:}
\textsf{magma}'s multiplication of a $(2n+1)\times(3n+1)$-matrix by a $(3n+1)\times(2n+1)$-matrix, that is, essentially all the linear algebra performed in variant~S (in practice, almost always in the cubic regime for the objects of interest);
{\bf Int:}
fully interpreted implementation of Strassen's product with cubic loop under a suitable threshold;
{\bf BZI and vdHI:}
variants of the implementations BZ and vdH (with evaluation-interpolation steps improved) in which \textsf{magma}'s product of matrices has been replaced with~Int.

\begin{table}[ht]
\begin{small}
\begin{center}
\setlength{\tabcolsep}{2.5pt}
\begin{tabular}{cc|rr|rrr|rrr}
$p$& & $p_1$& $p_1$& $p_2$& $p_2$& $p_2$&    0 &    0 &    0 \\
$k$& &    3 &    7 &    3 &    5 &    7 &    3 &    4 &    5 \\
\hline
LA & $\bigO\bigl(\MM(n)\bigr)$ &  4\% & 13\% & 17\% & 16\% & 39\% & 36\% & 41\% & 52\% \\
PP & $\bigO\bigl(n\,\sM(n)\bigr)$ & 13\% & 25\% & 23\% & 23\% & 18\% & 36\% & 33\% & 24\% \\
OM & $\bigO\bigl(n^2\bigr)$ & 38\% & 36\% & 30\% & 27\% & 11\% &  7\% &  6\% &  5\% \\
IO & $\bigO\bigl(n^2\bigr)$ & 46\% & 27\% & 30\% & 33\% & 32\% & 21\% & 20\% & 19\%
\end{tabular}
\caption{\label{table:percentages}Fraction of time spent in matrix product (LA), polynomial products (PP), other matrix operations (OM), and other interpreted operations (IO).}
\end{center}
\end{small}
\end{table}
\vskip-7pt

Comparing the columns Rec and, for instance, S, shows that linear algebra does not take the main part of the calculation time, although its theoretical complexity dominates.
In this regard, we have been very cautious in our implementation to avoid any interpreted quadratic loops.
Still, the result is that those quadratic tasks dominate the computation time.
Details are given in Table~\ref{table:percentages}.
The conclusion is that having implemented the algorithms in an interpreted language tends to parasitize the benchmarks.
For comparison sake, we have also added timings for variants BZI and vdHI that use an interpreted matrix product.
They both show the growth expected in theory, as well as the ratio from~8 to~96 announced in Table~\ref{table:MM}.

\section{Conclusions, future work}

Because of space limitation, various extensions could not be covered here.
More results
on the complexity of non-commutative multiplication of skew polynomials
will be
presented in an upcoming extended version~\cite{LongVersion}.
Topics like multiplication of skew
polynomials with unbalanced degrees and orders, or with sparse support, will be treated there.
The case of
rational (instead of polynomial) coefficients will also be considered.
The methods of this
article extend to multiplication of more general skew polynomials, in one or several variables,
including for instance $q$-recurrences and partial differential operators.

The constants in Table~\ref{table:MM} are all somewhat pessimistic.
Tighter bounds can be obtained by, on the one hand, relaxing the naive assumption~\eqref{eq:naive-abc}, on the other hand, taking advantage of the special shapes (banded, trapezoidal, etc) of the various matrices.

We also plan to provide a lower-level implementation.
Hopefully, the timings would then reflect the theoretical results even better and will be close to those of naked matrix products.

\smallskip\noindent{\bf Acknowledgments.} This work was supported in part by the French National Agency for Research (ANR Project ``Gecko'') and the Microsoft Research-INRIA Joint Centre.
We thank the three referees for their valuable comments.

\scriptsize
%\bibliographystyle{abbrv}
%\bibliography{issac}

\end{document}